\documentclass[submission,copyright,creativecommons]{eptcs}
\providecommand{\event}{QPL 2023}

\usepackage{iftex}

\ifpdf
  \usepackage{underscore}         
  \usepackage[T1]{fontenc}        
\else
  \usepackage{breakurl}           
\fi

\usepackage{enumerate,xspace}
\usepackage{amsmath,amssymb,wasysym}
\usepackage[all]{xy}
\usepackage{multicol}

\newtheorem{observation}{Remark}[section]
\newtheorem{lemma}[observation]{Lemma}  
\newtheorem{theorem}[observation]{Theorem}
\newtheorem{definition}[observation]{Definition}
\newtheorem{example}[observation]{Example}

\newtheorem{proposition}[observation]{Proposition} 
\newtheorem{corollary}[observation]{Corollary}

\title{Moore-Penrose Dagger Categories}
\author{Robin Cockett\thanks{Partially funded by NSERC.}
\institute{Department of Computer Science\\
University of Calgary, Canada}
\email{robin@ucalgary.ca}
\and
Jean-Simon Pacaud Lemay\thanks{Financially supported by a JSPS PDF, Award \#: P21746 and a ARC DECRA, Award \#: DE230100303}
\institute{School of Mathematical and Physical Sciences\\
Macquarie University, Australia}
\email{js.lemay@mq.edu.au}
}
\def\titlerunning{Moore-Penrose Dagger Categories}
\def\authorrunning{R. Cockett \& J.-S. P. Lemay}

\begin{document}
\maketitle

\begin{abstract}
The notion of a Moore-Penrose inverse (M-P inverse) was introduced by Moore in 1920 and rediscovered by Penrose in 1955.  The M-P inverse of a complex matrix is a special type of inverse which is unique, always exists, and can be computed using singular value decomposition. In a series of papers in the 1980s, Puystjens and Robinson studied M-P inverses more abstractly in the context of dagger categories. Despite the fact that dagger categories are now a fundamental notion in categorical quantum mechanics, the notion of a M-P inverse has not (to our knowledge) been revisited since their work.  One purpose of this paper is, thus, to renew the study of M-P inverses in dagger categories. 

Here we introduce the notion of a Moore-Penrose dagger category and provide many examples including complex matrices, finite Hilbert spaces, dagger groupoids, and inverse categories. We also introduce generalized versions of singular value decomposition, compact singular value decomposition, and polar decomposition for maps in a dagger category, and show how, having such a decomposition is equivalent to having M-P inverses. This allows us to provide precise characterizations of which maps have M-P inverses in a dagger idempotent complete category, a dagger kernel category with dagger biproducts (and negatives), and a dagger category with unique square roots. 
\end{abstract}

\section{Introduction}

The Moore-Penrose inverse of an $n \times m$ complex matrix $A$ is an $m\times n$ complex matrix $A^\circ$ such that: $AA^\circ A =A$, $A^\circ A A^\circ = A^\circ$, $(AA^\circ)^\dagger = AA^\circ$, and $(A^\circ A)^\dagger = A^\circ A$, where $\dagger$ is the conjugate transpose operator. For any complex matrix, its Moore-Penrose inverse exists, is unique, and can be computed using singular value decomposition -- see Example \ref{ex:MPmatC}. The Moore-Penrose inverse is named after E. H. Moore and R. Penrose. Moore first described the notion in 1920  in terms of orthogonal projectors \cite{moore1920reciprocal}. Without knowing about Moore's work, in 1955 Penrose described the notion using the identities above  \cite{penrose1955generalized}.  Curious readers can learn more about the fascinating history of the Moore-Penrose inverse and its origin in \cite{baksalary2021moore,ben2002moore}. Many useful -- and quite recent -- applications of the Moore-Penrose inverse in mathematics, physics, and computer science are described by Baksalary and Trenkler in \cite{baksalary2021moore}.

The Moore-Penrose inverse can be generalized to other contexts besides complex matrices. For example, one may consider the Moore-Penrose inverse of a matrix over an involutive ring. While the Moore-Penrose inverse may not always exist, for certain involutive rings it is possible to precisely characterize which matrices have Moore-Penrose inverses. One can also consider Moore-Penrose inverses in involutive semigroups, and in particular in $C^\ast$-algebras. It is also possible to define the notion of Moore-Penrose inverses for bounded linear operators between Hilbert spaces, and to characterize precisely which have a Moore-Penrose inverse. Following in this direction, one can in fact define the notion of a Moore-Penrose inverse for maps in \emph{dagger categories}. 

Selinger in \cite{selinger2007dagger} introduced the term ``dagger category'', based on the use in physics of the symbol $\dagger$ for conjugate transpose.  Dagger categories are simply categories equipped with an involution on maps (Def \ref{def:dagger}). In a dagger category, a Moore-Penrose inverse of a map $f: A \to B$ is a map in the reverse direction $f^\circ: B \to A$ satisfying the equations above (Def \ref{def:MP}). The existence and computations of Moore-Penrose inverses for maps in general dagger categories were studied by Puystjens and Robinson in a series of papers in the 1980s \cite{puystjens1981moore, puystjens1984moore, robinson1985ep, robinson1987generalized, puystjens1990symmetric}. Since Puystjens and Robinson's work, there does not appear to have been any further development of Moore-Penrose inverses in dagger categories. This, despite the fact that the theory of dagger categories itself has undergone significant development. Indeed, in the last decade, dagger categories have become a fundamental component of categorical quantum mechanics (see Heunen and Vicary's introductory level book on the subject \cite{heunen2019categories}). Therefore, it makes perfect sense to revisit Moore-Penrose inverses in the context of dagger categories. 

The main objective of this paper is to revisit and renew the study of Moore-Penrose inverses in dagger categories, in the hope that this will lead to new applications in categorical quantum mechanics and elsewhere. We shall apply techniques which have been developed since Puystjens and Robinson's work, such as dagger idempotent splitting and dagger kernels, to Moore-Penrose inverses. We also introduce and study the natural concept of a \textbf{Moore-Penrose dagger category}, which is a dagger category where every map has a Moore-Penrose inverse. We provide many examples of Moore-Penrose dagger categories including well-known ones, such as the category of complex matrices or finite-dimensional Hilbert spaces, and also various new ones, such as dagger groupoids and inverse categories. 

As was mentioned above, singular value decomposition can be used to compute Moore-Penrose inverses of complex matrices. In Section \ref{sec:SVD}, we introduce a generalized version of singular value decomposition for maps in a dagger category with dagger biproducts. Then, by using dagger kernels, we show how having a generalized singular value decomposition is equivalent to having a Moore-Penrose inverse (Thm \ref{thm:MPSVD}). Another way to compute the Moore-Penrose inverse is by using \emph{compact} singular value decomposition: this is often easier to compute than full singular value decomposition. In Section \ref{sec:CSVD}, we introduce a generalized version of compact singular value decomposition for maps in any dagger category and then prove that having a generalized compact singular value decomposition is equivalent to having a Moore-Penrose inverse when dagger idempotents split (Prop \ref{lem:CSVDMP}). Therefore, we obtain a precise characterization of maps that have a Moore-Penrose inverse in any dagger idempotent complete category (Thm \ref{thm:MPGCSVD}). Lastly in Section \ref{sec:CPD}, we give a novel application of Moore-Penrose inverses by introducing the notion of a Moore-Penrose polar decomposition, which captures precisely polar decomposition for complex matrices. 

\textbf{Acknowledgements:} The authors would like to thank Chris Heunen for useful discussions and support of this project, as well as thank Ben MacAdam and Cole Comfort for initial discussions on Moore-Penrose inverses and possible relations to restriction categories. The authors would also like to thank Masahito Hasegawa and RIMS at Kyoto University for helping fund research visits so that the authors could work together on this project. 

\section{Moore-Penrose Inverses}

In this section, we discuss Moore-Penrose inverses and some basic properties thereof. In addition, Moore-Penrose dagger categories are introduced and various examples are provided. To set up notation and terminology, we begin by quickly reviewing the basics of dagger categories. For a more in-depth introduction to dagger categories, we refer the reader to \cite{heunen2019categories}. For an arbitrary category $\mathbb{X}$, we denote objects by capital letters $A,B,X,Y$, etc. and maps by lowercase letters $f,g,h$, etc. Identity maps are denoted as $1_A: A \to A$.  Composition is written in \emph{diagrammatic order}, that is, the composition of a map $f: A \to B$ followed by $g: B \to C$ is denoted $fg: A \to C$. 

\begin{definition}\label{def:dagger} \cite[Def 2.32]{heunen2019categories} A \textbf{dagger} on a category $\mathbb{X}$ is a contravariant functor $(\_)^\dagger: \mathbb{X} \to \mathbb{X}$ which is the identity on objects and involutive. We refer to $f^\dagger$ as the \textbf{adjoint} of $f$. A \textbf{dagger category} is a pair $(\mathbb{X}, \dagger)$ consisting of a category $\mathbb{X}$ equipped with a dagger $\dagger$. 
\end{definition}

Concretely, a dagger category can be described as a category $\mathbb{X}$ where for each map $f: A \to B$, there is a chosen map of dual type $f^\dagger: B \to A$ such that $1^\dagger_A = 1_A$, $(fg)^\dagger = g^\dagger f^\dagger$, and $(f^\dagger)^\dagger = f$. Thus, $(\_)^\dagger$ is a contravariant functor which is, furthermore, an involution -- so the adjoint of the adjoint of $f$ is $f$ itself. It is important to note that a category $\mathbb{X}$ can have multiple different daggers. This means that a dagger on a category is structure which must be chosen. Examples of dagger categories can be found below. Here are some special maps in a dagger category: 

\begin{definition}\cite[Def 2.34]{heunen2019categories} In a dagger category $(\mathbb{X}, \dagger)$: 
\begin{enumerate}[{\em (i)}]
\item A map $s: A \to B$ is an \textbf{isometry} if $ss^\dagger = 1_A$;
\item A map $r: A \to B$ is a \textbf{coisometry} if $r^\dagger r=1_B$; 
\item A map $u: A \to B$ is a \textbf{unitary isomorphism} if $uu^\dagger=1_A$ and $u^\dagger u =1_B$;
\item A map $q: A \to B$ is a \textbf{partial isometry} if $qq^\dagger q = q$;
\item A map $h: A \to A$ is \textbf{self-adjoint} (or \textbf{Hermitian}) if $h^\dagger = h$;
\item A map $p: A \to A$ is \textbf{positive} if there exists a map $f: A \to X$ such that $p = ff^\dagger$;
\item A map $e: A \to A$ is a \textbf{$\dagger$-idempotent} if it self-adjoint and idempotent, that is, $e^\dagger= e$ and $ee=e$. 
\end{enumerate}
\end{definition}

This allows us to define the main concept of interest for this paper: 

\begin{definition} \label{def:MP} In a dagger category $(\mathbb{X}, \dagger)$, a \textbf{Moore-Penrose inverse} (M-P inverse) of a map $f: A \to B$ is a map $f^\circ: B \to A$ such that the following equalities hold: 
\begin{multicols}{4}
\begin{enumerate}[{\bf [MP.1]}]
\item $ff^\circ f = f$
\columnbreak
\item $f^\circ f f^\circ =f^\circ$
\columnbreak 
\item $(ff^\circ)^\dagger= ff^\circ$
\columnbreak
\item $(f^\circ f)^\dagger = f^\circ f$
\end{enumerate}
\end{multicols}
\noindent If $f$ has a M-P inverse, we say that $f$ is \textbf{Moore-Penrose invertible} (M-P invertible). A \textbf{Moore-Penrose dagger category} is a dagger category such that every map is M-P invertible. 
\end{definition}

{\bf [MP.1]} and {\bf [MP.2]} say that $f^\circ$ is a ``regular'' inverse of $f$, while {\bf [MP.3]} and {\bf [MP.4]} say that $ff^\circ$ and $f^\circ f$ are self-adjoint. This allows us to interpret $ff^\circ$ as the projection of the domain of $f$, while $f^\circ f$ is the projection of the range of $f$. Examples of Moore-Penrose dagger categories can be found below. However, before looking at examples, we state some basic results for M-P inverses. Most importantly, M-P inverses (if they exist) are unique: 

\begin{lemma} In a dagger category $(\mathbb{X}, \dagger)$, if a map $f: A \to B$ has a M-P inverse $f^\circ: B \to A$, then $f^\circ$ is the unique map which satisfies {\bf [MP.1]} to {\bf [MP.4]}.  
\end{lemma}

\begin{proof}
  Suppose that for a map $f: A \to B$, there exist maps $f^\circ: B \to A$ and $f^\bullet: B \to A$ which are both M-P inverses of $f$. Then we first compute that: 
\[  f^\bullet f  = f^\bullet f f^\circ f= (f^\bullet f)^\dagger (f^\circ f)^\dagger 
= f^\dagger (f^\bullet)^\dagger f^\dagger (f^\circ)^\dagger =  (f f^\bullet f)^\dagger  (f^\circ)^\dagger = f^\dagger (f^\circ)^\dagger = (f^\circ f)^\dagger = f^\circ f. \]
So $f^\bullet f = f^\circ f$ and, similarly, we can also compute that $f f^\bullet = f f^\circ$. This allows the observation that:
\[f^\bullet = f^\bullet f f^\bullet = f^\circ f f^\bullet = f^\circ f f^\circ = f^\circ\] 
So $f^\circ = f^\bullet$ and therefore Moore-Penrose inverses are unique.   
 \end{proof}

An important consequence of the above lemma is that, for a dagger category, being Moore-Penrose is a property rather than a structure. That said, it is important to note that a map can have a M-P inverse with respect to one dagger structure but fail to have one for another, see Example \ref{ex:MPmatCT}. Having a M-P inverse, has a number of consequences: 

\begin{lemma}\label{lem:MPinv1} In a dagger category $(\mathbb{X}, \dagger)$, if $f$ has a M-P inverse $f^\circ$ then: 
\begin{enumerate}[{\em (i)}]
\item $f^\circ$ is also M-P invertible where ${f^\circ}^\circ=f$; 
\item $f^\dagger$ is also M-P invertible where ${f^\dagger}^\circ = {f^\circ}^\dagger$;
\item \label{lem:MPfidem} $ff^\circ$ and $f^\circ f$ are $\dagger$-idempotents and M-P invertible where $(ff^\circ)^\circ = ff^\circ$ and $(f^\circ f)^\circ = f^\circ f$; 
\item \label{lem:MPdagidem} $ff^\dagger$ and $f^\dagger f$ are M-P invertible where $(ff^\dagger)^\circ = {f^\dagger}^\circ f^\circ$ and $(f^\dagger f)^\circ = f^\circ{f^\dagger}^\circ$; 
\item $ff^\circ= {f^\dagger}^\circ f^\dagger$ and $f^\circ f = f^\dagger {f^\dagger}^\circ$;
\item \label{lem:MPfeq} $f = f f^\dagger {f^\dagger}^\circ = {f^\dagger}^\circ f^\dagger f$;
\item \label{lem:MPfcirceq} $f^\circ = f^\circ {f^\dagger}^\circ f^\dagger = f^\dagger {f^\dagger}^\circ f^\circ$;
\item \label{lem:MPfdageq} $f^\dagger = f^\dagger f f^\circ = f^\circ f f^\dagger$; 
\item \label{lem:MPselfadjoint} If $f$ is self-adjoint, then $f^\circ$ is also self-adjoint (i.e. ${f^\circ}^\dagger=f^\circ$) and $f^\circ f = f f^\circ$;  
\item \label{lem:MPprisom2} If $f^\circ = f^\dagger$, then $f$ is a partial isometry. 
\end{enumerate}
\end{lemma}
\begin{proof} These are straightforward to check, so we leave them as an exercise for the reader.     \end{proof}

It is known that computing M-P inverses of complex matrices can be reduced to computing the M-P inverses of Hermitian positive semi-definite matrices. The same is true in dagger categories: 



\begin{lemma} In a dagger category $(\mathbb{X}, \dagger)$, for any map $f: A\to B$, the following are equivalent: 
\begin{enumerate}[{\em (i)}]
\item $f$ is M-P invertible;
\item $f^\dagger f$ is M-P invertible and $f (f^\dagger f)^\circ f^\dagger f = f$;
\item $ff^\dagger$ is M-P invertible and $f f^\dagger (ff^\dagger)^\circ f = f$
\end{enumerate}
Therefore $(\mathbb{X}, \dagger)$ is Moore-Penrose if and only if every map $f$ satisfies $(ii)$ or $(iii)$.    
\end{lemma}
\begin{proof} Lemma \ref{lem:MPinv1}.(\ref{lem:MPdagidem}) and (\ref{lem:MPfcirceq}) gives us $(i) \Rightarrow (ii)$ and $(i) \Rightarrow (iii)$. Conversely, if  $f^\dagger f$ (resp. $f f^\dagger$) is M-P invertible, then $(f^\dagger f)^\circ f^\dagger$ (resp. $f^\dagger (ff^\dagger)^\circ$) will always satisfy \textbf{[MP.2]}, \textbf{[MP.3]}, and \textbf{[MP.4]}. The extra assumption that $f (f^\dagger f)^\circ f^\dagger f = f$ (resp. $f f^\dagger (ff^\dagger)^\circ f = f$) is precisely \textbf{[MP.1]}. So we have that $f$ is M-P invertible, giving $(ii) \Rightarrow (i)$ and $(iii) \Rightarrow (i)$. 
\end{proof}


In any dagger category, there are some maps that always have M-P inverses: 

\begin{lemma}\label{lem:MPmaps} In a dagger category $(\mathbb{X}, \dagger)$: 
\begin{enumerate}[{\em (i)}]
\item Identity maps $1_A$ are M-P invertible where $1^\circ_A = 1_A$; 
\item If $f$ is an isomorphism, then $f$ is M-P invertible where $f^\circ = f^{\text{-}1}$;
\item \label{lem:MPparisom} If $f$ is a partial isometry or a (co)isometry or unitary, then $f$ is M-P invertible where $f^\circ=f^\dagger$;
\item \label{lem:MPidem} If $e$ is a $\dagger$-idempotent, then $e$ is M-P invertible where $e^\circ = e$;
\item If $p$ is a positive map such that there exists a M-P invertible map $f$ such that $p= ff^\dagger$, then $p$ is M-P invertible where $p^\circ = {f^\circ}^\dagger f^\circ$, and so $p^\circ$ is also positive;
\item If $p$ is a positive map and M-P invertible, then for any map $f$ such that $p=ff^\dagger$ and $pp^\circ f = f$, $f$ is also M-P invertible where $f^\circ = f^\dagger p^\circ$.
\end{enumerate}
\end{lemma}
\begin{proof} These are straightforward to check, so we leave them as an exercise for the reader.     \end{proof}

It is important to note that, in general, Moore-Penrose inverses are not compatible with composition. Indeed, even if $f$ and $g$ have M-P inverses, $fg$ might not have a M-P inverse and, even if it does, $(fg)^\circ$ is not necessarily equal to $g^\circ f^\circ$. Here are some conditions for when $(fg)^\circ = g^\circ f^\circ$ holds: 

\begin{lemma}\label{lem:MPcomp} In a dagger category $(\mathbb{X}, \dagger)$, if $f: A\to B$ and $g: B \to C$ are M-P invertible then:  
\begin{enumerate}[{\em (i)}]
\item $fg$ is M-P invertible with $(fg)^\circ = g^\circ f^\circ$
if and only if $f^\circ f g g^\circ$ and $g g^\circ f^\circ f$ are idempotent, and both
$f g g^\circ f^\circ = f^{\circ\dagger} g g^\circ f^\dagger$ and $g^\circ f^\circ f g = g^\dagger f^\circ f g^{\circ\dagger}$;
\item The following conditions\footnote{For complex matrices, the conditions of {\em (ii)} are equivalent to $(fg)^\circ = g^\circ f^\circ$ \cite[Sec 1.4 \& 1.5]{campbell2009generalized}. However, for general dagger categories it appears that the conditions in {\em (ii)} are sufficient -- but not necessary -- to obtain $(fg)^\circ = g^\circ f^\circ$.} are equivalent and imply $(fg)^\circ = g^\circ f^\circ$:
\begin{enumerate}[(a)]
\item $gg^\circ f^\circ f$, $fg g^\circ f^\circ$ and $g^\circ f^\circ fg$ are self-dual;
\item $gg^\dagger f^\circ f$ and $f^\dagger f g g^\circ$ are self-dual;
\item $f^\circ f g g^\dagger f^\dagger = gg^\dagger f^\dagger$ and $gg^\circ f^\dagger f g= f^\dagger f g$.
\end{enumerate}
\end{enumerate}
\end{lemma}
\begin{proof} These can be checked by lengthy and brute-force calculations.  \end{proof}

Here are some examples of Moore-Penrose dagger categories, as well as some non-examples but where we can still fully characterize the M-P invertible maps: 

\begin{example}\label{ex:MPmatC} Let $\mathbb{C}$ be the field of complex numbers and let $\mathsf{MAT}(\mathbb{C})$ be the category whose objects are natural numbers $n \in \mathbb{N}$ and where a map $A: n \to m$ is an $n \times m$ complex matrix. $(\mathsf{MAT}(\mathbb{C}), \dagger)$ is a dagger category where $\dagger$ is the conjugate transpose operator, $A^\dagger(i,j) = \overline{A(j,i)}$. Furthermore, $(\mathsf{MAT}(\mathbb{C}), \dagger)$ is also a Moore-Penrose dagger category where the M-P inverse of a matrix can be constructed from its singular value decomposition (SVD). For a $n \times m$ $\mathbb{C}$-matrix $A$, let $d_1, \hdots, d_k$ be the non-zero singular values of $A$, so $d_i \in \mathbb{R}$ with $d_i >0$, and $k \leq \mathsf{min}(n,m)$. Then there exists a unitary $n \times n$ matrix $U$ and a unitary $m \times m$ matrix $V$ such that: 
\begin{align*}
    A = U \begin{bmatrix} D & 0 \\
    0 & 0 
    \end{bmatrix}_{n \times m} V^\dagger && \text{where $D$ is the diagonal $k \times k$ matrix } D= \begin{bmatrix}
    d_{1} & \hdots & 0 \\
   \vdots  & \ddots & \vdots \\
   0 & \hdots & d_{k}
  \end{bmatrix}
\end{align*}
Then the M-P inverse of $A$ is the $m \times n$ matrix $A^\circ$ defined as follows: 
\begin{align*}
    A^\circ = V \begin{bmatrix} D^{\text{-}1} & 0 \\
    0 & 0 
    \end{bmatrix}_{m \times n}  U^\dagger && \text{where $D^{\text{-}1}$ is the diagonal $k \times k$ matrix } D^{\text{-}1} = \begin{bmatrix}
    \frac{1}{d_{1}} & \hdots & 0 \\
   \vdots  & \ddots & \vdots \\
   0 & \hdots & \frac{1}{d_{k}} 
  \end{bmatrix}
\end{align*}
Since M-P inverses are unique, the construction does not depend on the choice of SVD.
\end{example}

\begin{example}\label{ex:MPmatCT}  On the other hand, $\mathsf{MAT}(\mathbb{C})$ has another dagger given instead simply by the transpose operator, $A^\mathsf{T}(i,j) = A(j,i)$. However, the dagger category $(\mathsf{MAT}(\mathbb{C}), \mathsf{T})$ is not Moore-Penrose. For example, the matrix $\begin{bmatrix} i & 1 \end{bmatrix}$ does not have a M-P inverse with respect to the transpose. If it did, one can obtain the contradiction that $i=0$, which we leave as an exercise for the reader. 
\end{example}

\begin{example}\label{ex:mat*} Recall that an involutive ring is a ring $R$ equipped with a unary operation $\ast$, called the involution, such that $(x+y)^\ast = x^\ast + y^\ast$, and $(xy)^\ast = y^\ast x^\ast$, and ${x^\ast}^\ast = x$. Let $\mathsf{MAT}(R)$ be the category of matrices over $R$, that is, the category whose objects are natural numbers $n \in \mathbb{N}$ and where a map $A: n \to m$ is an $n \times m$ matrix $A$ with coefficients in $R$. Then $(\mathsf{MAT}(R), \dagger)$ is a dagger category where $\dagger$ is given by the involution transpose operator, that is, $A^\dagger(i,j) = A(j,i)^\ast$. In general $(\mathsf{MAT}(R), \dagger)$ will not necessarily be Moore-Penrose. However, in certain cases, it is possible to precisely characterize which $R$-matrices do have a M-P inverse. For example, if $R$ is an involutive field, then an $R$-matrix $A$ has a M-P inverse if and only if $\mathsf{rank}(AA^\dagger)=\mathsf{rank}(A) = \mathsf{rank}(A^\dagger A)$ \cite[Thm 1]{pearl1968generalized}. Necessary and sufficient conditions for when an $R$-matrix has a M-P inverse have also been described in the case when $R$ is an integral domain \cite{bapat1990generalized}, a commutative ring \cite{bapat1992moore}, or even a semi-simple artinian ring \cite{huylebrouck1984moore}.
\end{example}

\begin{example} Let $\mathsf{HILB}$ be the category of (complex) Hilbert spaces and bounded linear operators between them. Then $(\mathsf{HILB}, \dagger)$ is a dagger category where the dagger is given by the adjoint, that is, for a bounded linear operator $f: H_1 \to H_2$, $f^\dagger: H_2 \to H_1$ is the unique bounded linear operator such that $\langle f(x) \vert y \rangle  =\langle x \vert f^\dagger(y) \rangle$ for all $x \in H_1$ and $y \in H_2$. $(\mathsf{HILB}, \dagger)$ is not Moore-Penrose but there is a characterization of the M-P invertible maps: a bounded linear operator is M-P invertible if and only if its range is closed \cite[Thm 2.4]{hagen2000c}. Explicitly, for a bounded linear map $f: H_1 \to H_2$, let $\mathsf{Ker}(f) \subseteq H_1$ be its kernel and $\mathsf{im}(f) \subseteq H_2$ be its range, and let $\mathsf{Ker}(f)^\perp$ and $\mathsf{im}(f)^\perp$ be their orthogonal complements. If $\mathsf{im}(f)$ is closed, then we have that $H_2 = \mathsf{im}(f) \oplus \mathsf{im}(f)^\perp$ and also that $f\vert_{\mathsf{Ker}(f)^\perp}: \mathsf{Ker}(f)^\perp \to \mathsf{im}(f)$ is a bounded linear isomorphism. Then define the M-P inverse $f^\circ: H_2 \to H_1$ as $f^\circ(y) = f^{-1}\vert_{\mathsf{Ker}(f)^\perp}(y)$ for $y \in \mathsf{im}(f)$ and $f^\circ(y) = 0$ for $y \in \mathsf{im}(f)^\perp$. For more details, see \cite[Ex 2.16]{hagen2000c}. Now let $\mathsf{FHILB}$ be the subcategory of finite dimensional Hilbert spaces. Then $(\mathsf{FHILB}, \dagger)$ is also a dagger category and it is well known that $(\mathsf{FHILB}, \dagger) \simeq (\mathsf{MAT}(\mathbb{C}), \dagger)$. As such, $(\mathsf{FHILB}, \dagger)$ is also a Moore-Penrose dagger category where we this time use SVD on linear operators to construct the M-P inverse. So let $H_1$ be a Hilbert space of dimension $n$ and $H_2$ a Hilbert space of dimension $m$. Then for any linear operator $f: H_1 \to H_2$, if $d_1, \hdots, d_k \in \mathbb{R}$ are the non-zero singular values of $f$ (so $k \leq \mathsf{min}(n,m)$), then there exists orthonormal bases $u_i \in H_1$ and $v_j \in H_2$ such that $f(x) = \sum^k_{i=1} d_i \langle u_i \vert x \rangle v_i$ for all $x \in H_1$. Then $f^\circ: H_2 \to H_1$ is defined as follows $f^\circ(y) := \sum^k_{i=1} \frac{1}{d_i} \langle v_i \vert y \rangle u_i$.  
\end{example}

\begin{example} \label{ex:field} Any field gives a simple example of a Moore-Penrose dagger category. So let $k$ be a field, and let $\bullet_k$ be the category with one object and whose maps are elements of $k$, where composition is given by the multiplication and the identity map is the unit of $k$. Then $(\bullet_k, \dagger)$ is a Moore-Penrose dagger category where for all $x \in k$, $x^\dagger = x$ and $x^\circ = x^{\text{-}1}$ if $x \neq 0$ or $x^\circ =0$ if $x=0$. In fact, a Moore-Penrose dagger category with only one object is precisely a $\ast$-regular monoid \cite{drazin1979regular}. 
\end{example}

\begin{example} \label{ex:REL} Let $\mathsf{REL}$ be the category of sets and relations, that is, the category whose objects are sets and where a map $R: X \to Y$ is a subset $R \subseteq X \times Y$. $(\mathsf{REL}, \dagger)$ is a dagger category where $\dagger$ is given by the converse relation, that is, $(y,x) \in R^\dagger \subseteq Y \times X$ if and only if $(x,y) \in R \subseteq X \times Y$. While $(\mathsf{REL}, \dagger)$ is not a Moore-Penrose dagger category, it turns out that the M-P invertible maps are precisely the partial isometries (which recall by Lemma \ref{lem:MPmaps}.(\ref{lem:MPparisom}) always have M-P inverses). A partial isometry in $(\mathsf{REL}, \dagger)$ is a difunctional relation \cite[Def 1]{gumm2014coalgebraic}, which is a relation $R \subseteq X \times Y$ which satisfies that if $(x,b), (a,b)$ and $(a,y) \in R$, then $(x,y) \in R$. It was previously observed that a relation between \emph{finite} sets has M-P inverse if and only if it was a difunctional relation/partial isometry -- since relations between finite sets correspond to Boolean matrices, and Boolean matrices with M-P inverses were fully characterized in \cite[Thm 4.3]{rao1975generalized}. From this, it is not difficult to see that this can be extended to relations between arbitrary sets. Thus, in $(\mathsf{REL}, \dagger)$, $R \subseteq X \times Y$ has a M-P inverse if and only if $R$ is a difunctional relation/partial isometry, which in this case means that the M-P inverse is the converse relations $R^\circ = R^\dagger \subseteq Y \times X$. In fact, the same is true for \emph{allegories}. Briefly, an \textbf{allegory} \cite[Chap 2]{freyd1990categories} is a dagger category $(\mathbb{X}, \dagger)$ which is poset enriched and has meets, so in particular each homset $\mathbb{X}(A,B)$ is a poset with order $\leq$ and binary meets $\cap$, and such that the modular law $fg \cap h \leq (f \cap hg^\dagger)g$ holds. Well-known examples of allegories include $(\mathsf{REL}, \dagger)$ and more generally the category of relations of a regular category \cite[Sec 2.111]{freyd1990categories}. From the modular law, it follows that every map $f$ in an allegory $(\mathbb{X}, \dagger)$ satisfies $f \leq f f^\dagger f$ \cite[Sec 2.112]{freyd1990categories}. Therefore, if $f$ has a M-P inverse, using Lemma \ref{lem:MPinv1}.(\ref{lem:MPfcirceq}) and (\ref{lem:MPfdageq}), we easily compute that: 
\[ f^\dagger = f^\circ f f^\dagger \leq f^\circ {f^\circ}^\dagger f^\circ  f f^\dagger = f^\circ  f^{\circ \dagger} f^\dagger = f^\circ \]
\[ f^\circ = f^\circ {f^\circ}^\dagger f^\dagger \leq f^\circ {f^\circ}^\dagger f^\dagger f f^\dagger =  f^\circ f f^\dagger = f^\dagger.\]
So we conclude that $f^\circ = f^\dagger$, and so by Lemma \ref{lem:MPinv1}.(\ref{lem:MPprisom2}), $f$ is a partial isometry. Thus, a map $f$ in an allegory $(\mathbb{X}, \dagger)$ has a M-P inverse if and only if $f$ is a partial isometry, which means that its M-P inverse is its adjoint $f^\circ = f^\dagger$. 
\end{example}

\begin{example}\label{ex:daggergroupoid} A \textbf{dagger groupoid} is a dagger category $(\mathbb{X}, \dagger)$ where every map in $\mathbb{X}$ is an isomorphism (though not necessarily a unitary). Every dagger groupoid $(\mathbb{X}, \dagger)$ is a Moore-Penrose dagger category where $f^\circ = f^{\text{-}1}$. In particular, from any dagger category, we can always construct a dagger groupoid via its subcategory of isomorphisms. So for any category $\mathbb{X}$, let $\mathbb{X}_{\mathsf{iso}}$ be the subcategory of isomorphisms of $\mathbb{X}$. If $(\mathbb{X}, \dagger)$ is a dagger category, then $(\mathbb{X}_{\mathsf{iso}}, \dagger)$ is a dagger groupoid since if $f$ is an isomorphism, then so is $f^\dagger$ with inverse ${f^\dagger}^{\text{-}1} := {f^{\text{-}1}}^\dagger$. Therefore $(\mathbb{X}_{\mathsf{iso}}, \dagger)$ is a Moore-Penrose dagger category. 
\end{example}
 
\begin{example}\label{ex:inverse} An \textbf{inverse category} \cite[Sec 2.3.2]{cockett2002restriction} is a dagger category $(\mathbb{X}, \dagger)$ where $ff^\dagger f = f$ for all maps $f$ and $ff^\dagger gg^\dagger = g g^\dagger ff^\dagger$ for all parallel maps $f$ and $g$. Inverse categories play an important role in the theory of restriction categories \cite{cockett2002restriction}, since the subcategory of partial isomorphisms of a restriction category is an inverse category. Every inverse category $(\mathbb{X}, \dagger)$ is a Moore-Penrose dagger category where the M-P inverse of $f$ is its adjoint $f^\circ = f^\dagger$ (since every map in an inverse category is a partial isometry by definition). So in particular, for any restriction category, its subcategory of partial isomorphisms is a Moore-Penrose dagger category. As a concrete example, let $\mathsf{PINJ}$ be the category of sets and partial injections, which is the subcategory of partial isomorphisms of the restriction category of sets and partial functions. Then $(\mathsf{PINJ}, \dagger)$ is an inverse category where for a partial injection $f: X \to Y$, $f^\dagger: Y \to X$ is defined as $f^\dagger(y) = x$ if $f(x)=y$ and is undefined otherwise. 
\end{example}

\begin{example} If $(\mathbb{X}_1, \dagger_1)$ and $(\mathbb{X}_2, \dagger_2)$ are both Moore-Penrose dagger categories, then their product $(\mathbb{X}_1 \times \mathbb{X}_2, \dagger_1 \times \dagger_2)$ is also a Moore-Penrose dagger category. In particular, we can combine Example \ref{ex:field} and Example \ref{ex:inverse}. So if $(\mathbb{X}, \dagger)$ is an inverse category and $k$ is a field, let $\mathbb{X}_k$ be the category whose objects are those of $\mathbb{X}$ but whose maps are pairs $(f,x)$ consisting of a map $f$ in $\mathbb{X}$ and an element $x \in k$, so we may think of $x$ as adding a weight or a cost to $f$. Then $(\mathbb{X}_k, \dagger)$ is a Moore-Penrose dagger category where $(f,x)^\dagger = (f^\dagger,x)$ and $(f,x)^\circ= (f^\dagger, x^\circ)$. 
\end{example}

\section{Compact Singular Value Decomposition}\label{sec:CSVD}

In Example \ref{ex:MPmatC}, we explained how to construct the M-P inverse of a complex matrix using SVD. However, there is an alternative way to construct the M-P inverse using \emph{compact} singular value decomposition (CSVD). This decomposition tells us that for any $n\times m$ complex matrix, $A$, again with singular values $d_1,\hdots,d_k$ and associated diagonal matrix $D$, there exists an $n\times k$ matrix $R$ and an $m \times k$ matrix $S$ such that $A=RDS^\dagger$ and $R^\dagger R= S^\dagger S= I_k$. The decomposition allows one to construct the M-P inverse as $A^\circ := S D^{-1} R^\dagger$. In dagger categorical terms, $R$ and $S$ are coisometries, and $D$ is an isomorphism\footnote{The fact that $D$ is a diagonal matrix of singular values is not relevant to this way of constructing the M-P inverse.}. Thus, generalized CSVD in an arbitrary dagger category is a factorization into a coisometry, followed by an isomorphism, followed by an isometry. We shall discuss the generalization of CSVD for dagger categories before discussing SVD because generalizing SVD requires dagger biproducts and dagger kernels, while generalizing CSVD can be explained without introducing further structure.  

This generalized CSVD not only provides a simple way of computing M-P inverses, but is also directly related to the splitting of dagger idempotents, an important dagger category concept that was introduced by Selinger in \cite{selinger2008idempotents}. Generalized CSVD allows us to precisely characterize the M-P invertible maps in dagger categories which are dagger idempotent complete. Furthermore, the dagger idempotent splitting completion leads us to an important reinterpretation of M-P inverses as being {\em actual\/} inverses between dagger idempotents. As such, we begin this section by discussing the relationship between M-P inverses and dagger idempotent splitting. 

\begin{definition} \cite[Def 3.6]{selinger2008idempotents} In a dagger category $(\mathbb{X}, \dagger)$, a \textbf{dagger idempotent} $e: A \to A$ is an idempotent which is self-adjoint, $ee = e = e^\dagger$.  A dagger idempotent is said to \textbf{$\dagger$-split} if there exists a map $r: A \to X$ such that $rr^\dagger = e$ and $r^\dagger r=1_X$ (so $r$ is a coisometry). A \textbf{dagger idempotent complete category} is a dagger category $(\mathbb{X}, \dagger)$ such that all $\dagger$-idempotents $\dagger$-split. 
\end{definition}

In Lemma \ref{lem:MPinv1}.(\ref{lem:MPfidem}), we saw that in any dagger category $(\mathbb{X}, \dagger)$, if a map $f$ has a M-P inverse, then $ff^\circ$ and $f^\circ f$ were both $\dagger$-idempotents. As such, we may ask these $\dagger$-idempotents to also be $\dagger$-split: 

\begin{definition} In a dagger category $(\mathbb{X}, \dagger)$, a map $f$ is \textbf{Moore-Penrose split} (M-P split) if $f$ has a M-P inverse $f^\circ$ and the $\dagger$-idempotents $f f^\circ$ and $f^\circ f$ $\dagger$-split. A Moore-Penrose category in which all maps are M-P split is said to be \textbf{Moore-Penrose complete}. 
\end{definition}

A dagger category which is Moore-Penrose complete is the same thing as a Moore-Penrose category in which {\em all\/} dagger idempotents split:  

\begin{proposition} A dagger category $(\mathbb{X}, \dagger)$ is Moore-Penrose complete if and only if $(\mathbb{X}, \dagger)$ is dagger idempotent complete and Moore-Penrose.
\end{proposition}
\begin{proof} The $\Leftarrow$ direction is immediate by definition. For the $\Rightarrow$ direction, suppose that $(\mathbb{X}, \dagger)$ is Moore-Penrose complete. By definition, this means every map has a M-P inverse, so $(\mathbb{X}, \dagger)$ is indeed Moore-Penrose. Now let $e: A \to A$ be a $\dagger$-idempotent. By Lemma \ref{lem:MPmaps}.(\ref{lem:MPidem}), $e$ is its own M-P inverse, so $e^\circ = e$, and therefore $e^\circ e = e = e e^\circ$. However, by assumption, $e$ is M-P split, which therefore implies that $e$ is $\dagger$-split. So $(\mathbb{X}, \dagger)$ is indeed $\dagger$-idempotent complete.
\end{proof}

We will now explain how every Moore-Penrose dagger category embeds into a Moore-Penrose complete dagger category. Let us first review how every dagger category embeds into a dagger idempotent complete category via the dagger version of the idempotent splitting completion, also called the dagger Karoubi envelope \cite[Def 3.13]{selinger2008idempotents}. So for a dagger category $(\mathbb{X}, \dagger)$, define the dagger category $({\sf Split}_\dagger(\mathbb{X}), \dagger)$ whose objects are pairs $(A,e)$ consisting of an object $A$ and a $\dagger$-idempotent $e: A \to A$ in $(\mathbb{X}, \dagger)$, and whose maps $f: (A_1,e_1) \to (A_2, e_2)$ in are maps $f: A_1 \to A_2$ in $\mathbb{X}$ such that $e_1fe_2 = f$ (or equivalently $e_1f = f= fe_2$). Composition in ${\sf Split}_\dagger(\mathbb{X})$ is defined as in $\mathbb{X}$, while identity maps $1_{(A,e)}: (A,e) \to (A,e)$ are defined as $1_{(A,e)} := e$. Lastly, the dagger of $({\sf Split}_\dagger(\mathbb{X}), \dagger)$ is defined as in $(\mathbb{X}, \dagger)$, and furthermore $({\sf Split}_\dagger(\mathbb{X}), \dagger)$ is a dagger idempotent complete category \cite[Prop 3.12]{selinger2008idempotents}. There is also an embedding $\mathcal{I}: (\mathbb{X}, \dagger) \to ({\sf Split}_\dagger(\mathbb{X}), \dagger)$  which is defined on objects as $\mathcal{I}(A) = (A, 1_A)$ and on maps as $\mathcal{I}(f)=f$. 

\begin{lemma}\label{lem:Karoubi} Let $(\mathbb{X}, \dagger)$ be a Moore-Penrose dagger category. Then $({\sf Split}_\dagger(\mathbb{X}), \dagger)$ is a Moore-Penrose complete category. 
\end{lemma}
\begin{proof} Let $f: (A,e) \to (B, e')$ be a map in $({\sf Split}_\dagger(\mathbb{X}), \dagger)$. Since composition and the dagger of $({\sf Split}_\dagger(\mathbb{X}), \dagger)$ are the same as in $(\mathbb{X}, \dagger)$, it suffices to show that $f^\circ: B \to A$ is also a map of type ${(B,e') \to (A, e)}$ in $({\sf Split}_\dagger(\mathbb{X}), \dagger)$. So we must show that $e' f^\circ e = f^\circ$. To do so we use Lemma \ref{lem:MPinv1}.(\ref{lem:MPfcirceq}) and that $f^\dagger: (A_2,e_2) \to (A_1, e_1)$ is also a map in $({\sf Split}_\dagger(\mathbb{X}), \dagger)$: 
\begin{gather*}
    e' f^\circ e = e' f^\circ {f^\dagger}^\circ f^\dagger e =  e' f^\circ {f^\dagger}^\circ f^\dagger = e' f^\circ = e' f^\dagger {f^\dagger}^\circ f^\circ = f^\dagger {f^\dagger}^\circ f^\circ = f^\circ
\end{gather*}
So $f^\circ: (B,e') \to (A, e)$ is a map in $({\sf Split}_\dagger(\mathbb{X}), \dagger)$. 
\end{proof}

We are now ready to discuss a generalization of CSVD in an arbitrary dagger category, and show that having a generalized CSVD is equivalent to being M-P split. 

\begin{definition} In a dagger category, a \textbf{generalized compact singular value decomposition} (GCSVD) of a map $f: A \to B$ is a triple $(r: A \to X, d: X \to Y, s: Y \to B)$, where $r$ is a coisometry, $d$ is an isomorphism, and $s$ is an isometry, such that $f = r d s$. 
\end{definition}

\begin{lemma}\label{lem:CSVDunique} In a dagger category $(\mathbb{X}, \dagger)$, if the two triples $({r_1: A \to X_1}, {d_1: X_1 \to Y_1}, s_1: Y_1 \to B)$ and $(r_2: A \to X_2, d_2: X_2 \to Y_2, s_2: Y_2 \to B)$ are GCSVDs of $f: A \to B$, then there exist unique unitary maps $u: X_1 \to X_2$ and $v: Y_1 \to Y_2$ such that $r_1u =r_2$, $d_1 v= u d_2$, and $s_1 = v s_2$. 
\end{lemma}
\begin{proof} Define $u$ and $v$ as the composites, $u:= r_1^\dagger r_2$ and $v:= s_1 s_2^\dagger$. The necessary identities are checked via some straightforward diagram chasing. \end{proof}

In order to show that having a GCSVD is equivalent to being M-P split, it will be useful to first observe that maps with M-P inverses in the base dagger category are actual isomorphisms in the dagger idempotent splitting completion: 

\begin{lemma}\label{lem:MPisosplit} A map $f: A \to B$ in a dagger category $(\mathbb{X}, \dagger)$ has a M-P inverse if and only if there exists $\dagger$-idempotents $e_1: A \to A$ and $e_2: B \to B$ such that $f: (A,e_1) \to (B, e_2)$ is an isomorphism in $({\sf Split}_\dagger(\mathbb{X}), \dagger)$. Explicitly:
\begin{enumerate}[{\em (i)}]
\item If $f: A \to B$ has a M-P inverse $f^\circ: B \to A$, then ${f: (A, ff^\circ) \to (A, f^\circ f)}$ is an isomorphism in $({\sf Split}_\dagger(\mathbb{X}), \dagger)$ with inverse $f^\circ: (A, f^\circ f) \to (A, f f^\circ )$;
\item If $f: (A,e_1) \to (B, e_2)$ is an isomorphism in $({\sf Split}_\dagger(\mathbb{X}), \dagger)$ with inverse ${f^\circ:  (B, e_2) \to (A,e_1)}$, then $f$ is M-P invertible in $(\mathbb{X}, \dagger)$ with M-P inverse $f^\circ$. 
\end{enumerate}
\end{lemma}
\begin{proof} To start, let us explicitly spell out what it means for $f: (A,e_1) \to (B, e_2)$ to be an isomorphism in $({\sf Split}_\dagger(\mathbb{X}), \dagger)$. Firstly, we need that $e_1 f e_2 = f$ (or equivalently $e_1 f =f = fe_2)$. Secondly, we also need a map $g: (B, e_2) \to (A, e_1)$ in $({\sf Split}_\dagger(\mathbb{X}), \dagger)$, so $e_2 g e_1 = g$ (or equivalently $e_2 g = g = g e_1$), and such that $fg = 1_{ (A,e_1)} = e_1$ and $gf = 1_{ (B,e_2)} = e_2$. 

Suppose that $f: A \to B$ has a M-P inverse $f^\circ: B \to A$. By Lemma \ref{lem:MPinv1}.(\ref{lem:MPfidem}), $(A, ff^\circ)$ and $(A, f^\circ f)$ are well-defined objects in $({\sf Split}_\dagger(\mathbb{X}), \dagger)$. On the other hand, by \textbf{[MP.1]} and \textbf{[MP.2]}, it is easy to check that $f: (A, ff^\circ) \to (A, f^\circ f)$ and $f^\circ: (A, f^\circ f) \to (A, f f^\circ )$ are well-defined maps in $({\sf Split}_\dagger(\mathbb{X}), \dagger)$. Lastly, by definition we have that $ff^\circ = 1_{ (A,ff^\circ)}$ and $f^\circ f = 1_{ (A,f^\circ f)}$. Thus we conclude that ${f: (A, ff^\circ) \to (A, f^\circ f)}$ is an isomorphism in $({\sf Split}_\dagger(\mathbb{X}), \dagger)$. 

Conversely, suppose that $f: (A,e_1) \to (B, e_2)$ is an isomorphism in $({\sf Split}_\dagger(\mathbb{X}), \dagger)$ with inverse $f^\circ:  (B, e_2) \to (A,e_1)$. In particular, this implies that $ff^\circ = e_2$ and $ff^\circ = e_1$. So $ff^\circ$ and $f^\circ f$ are both $\dagger$-idempotents, thus \textbf{[MP.3]} and \textbf{[MP.4]} hold. By the assumed properties of maps in $({\sf Split}_\dagger(\mathbb{X}), \dagger)$, we have that $ff^\circ f = e_1 f = f$ and $f^\circ f f^\circ = e_2 f^\circ = f^\circ$, and so \textbf{[MP.1]} and \textbf{[MP.2]} hold. Therefore, we conclude that $f$ is M-P invertible with M-P inverse $f^\circ$. 
\end{proof}

\begin{corollary}\label{cor:MPisosplit} A map $f: A \to B$ in a dagger category $(\mathbb{X}, \dagger)$ is M-P split if and only if there exists $\dagger$-split $\dagger$-idempotents $e_1: A \to A$ and $e_2: B \to B$ such that $f: (A,e_1) \to (B, e_2)$ is an isomorphism in $({\sf Split}_\dagger(\mathbb{X}), \dagger)$. 
\end{corollary}
\begin{proof} Suppose that $f: A \to B$ is M-P split with M-P inverse $f^\circ: B \to A$. By definition $ff^\circ$ and $f^\circ f$ are $\dagger$-split $\dagger$-idempotents and by Lemma \ref{lem:MPisosplit}, ${f: (A, ff^\circ) \to (A, f^\circ f)}$ is an isomorphism in $({\sf Split}_\dagger(\mathbb{X}), \dagger)$. Conversely, suppose that $e_1: A \to A$ and $e_2: B \to B$ are $\dagger$-idempotents that $\dagger$-split via the coisometry $r: A \to X$ and isometry $s: Y \to B$ respectively, and also that $f: (A,e_1) \to (B, e_2)$ is an isomorphism in $({\sf Split}_\dagger(\mathbb{X}), \dagger)$ with inverse ${f^\circ:  (B, e_2) \to (A,e_1)}$. Then by Lemma \ref{lem:MPisosplit}, $f^\circ$ is the M-P inverse of $f$, and by assumption we also have that $f f^\circ = e_1$ and $f^\circ f = e_2$. So $f f^\circ$ and $f^\circ f$ are $\dagger$-split, and therefore we conclude that $f$ is M-P split. \hfill \end{proof}

We may now state the main result of this section: 

\begin{proposition}\label{lem:CSVDMP} In a dagger category $(\mathbb{X}, \dagger)$, a map $f$ has a GCSVD if and only if $f$ is M-P split. 
\end{proposition}
\begin{proof} Suppose that $f: A \to B$ has a GCSVD $(r: A \to X, d: X \to Y, s: Y \to B)$. Define $f^\circ := s^\dagger d^{\text{-}1} r^\dagger$. 
First note that $rr^\dagger$ and $s^\dagger s$ are $\dagger$-split $\dagger$-idempotents, so $(A, rr^\dagger)$ and $(B, s^\dagger s)$ are well-defined objects in $({\sf Split}_\dagger(\mathbb{X}), \dagger)$. We then compute that: 
\begin{gather*} rr^\dagger f s^\dagger s =  rr^\dagger r d s s^\dagger s = rds = f \\
s^\dagger s f^\circ rr^\dagger =  s^\dagger s  s^\dagger d^{\text{-}1} r^\dagger rr^\dagger =  s^\dagger d^{\text{-}1} r^\dagger = f^\circ 
\end{gather*}
So $f: (A, rr^\dagger) \to (B, s^\dagger s)$ and $f^\circ: (B, s^\dagger s) \to (A, rr^\dagger)$ are maps in $({\sf Split}_\dagger(\mathbb{X}), \dagger)$. Furthermore, we can also compute that: 
\begin{gather*} ff^\circ =   rdss^\dagger d^{\text{-}1} r^\dagger = r d d^{\text{-}1} r^\dagger = rr^\dagger = 1_{(A, rr^\dagger)}  \\
 f^\circ f = s^\dagger d^{\text{-}1} r^\dagger rds  =  s^\dagger d^{\text{-}1} ds = s^\dagger s = 1_{(B, s^\dagger s)}  
\end{gather*}
Therefore, ${f: (A, ff^\circ) \to (A, f^\circ f)}$ is an isomorphism with inverse $f^\circ: (B, s^\dagger s) \to (A, rr^\dagger)$. So by Corollary  \ref{cor:MPisosplit}, $f$ is M-P split with M-P inverse $f^\circ$. 

Conversely, suppose that $f: A \to B$ is M-P split, where $f f^\circ$ and $f^\circ f$ both $\dagger$-split via, respectively, the coisometry $r: A \to X$ and isometry $s: Y \to B$. Now define $d: X \to Y$ as the composite $d := r^\dagger f s^\dagger$. We then immediately have  $f = rd s$. So it remains to show that $d$ is an isomorphism. So define $d^{\text{-}1}: Y \to X$ as the composite $d^{\text{-}1} = s^\dagger f^\circ r$. We compute that:  
\begin{gather*}
    d d^{\text{-}1} \!\!= r^\dagger f s^\dagger s f^\circ r = r^\dagger f f^\circ f f^\circ r = r^\dagger r r^\dagger r r^\dagger r = 1_X \\
      d^{\text{-}1}d \!= s f^\circ rr^\dagger f s^\dagger =  s f^\circ f f^\circ f s^\dagger = s s^\dagger s s^\dagger s s^\dagger = 1_Y 
\end{gather*}
Therefore, $(r: A \to X, d: X \to Y, s: Y \to B)$ is a CSVD of $f$. 
\end{proof}

We can now precisely characterize M-P invertible maps in a dagger idempotent complete category: 

\begin{theorem}\label{thm:MPGCSVD} In a dagger idempotent complete category, a map is M-P invertible if and only if it has a GCSVD. 
\end{theorem}


\begin{corollary} A dagger category is Moore-Penrose complete if and only if every map has a GCSVD. 
\end{corollary}

Observe that Lemma \ref{lem:Karoubi} tells us that every Moore-Penrose dagger category embeds into a dagger category where every map has a GCSVD. 





\section{Singular Value Decomposition}\label{sec:SVD}

The objective of this section is to generalize SVD for maps in a dagger category in such a way that we may compute M-P inverses in the same way that was done in Example \ref{ex:MPmatC}. So generalized SVD can be described as a special factorization in terms of two unitaries and an isomorphism. However, in order to describe the middle component as a square matrix with the isomorphism in the top corner and zeroes everywhere else, we need to work in a setting with \emph{dagger biproducts}. It is worth mentioning that in \cite{puystjens1984moore}, Puystjens and Robinson do discuss how the existence of a M-P inverse for a map with an epic-monic factorization is essentially equivalent to a factorization via dagger biproducts. Here, we drop the epic-monic factorization requirement, which allows us to provide a story of how M-P inverses are equivalent to a dagger biproduct factorization which more closely resembles the generalized version of SVD. 

Let us begin by quickly recalling the definition of dagger biproducts. For a refresher on biproducts and zero objects, we refer the reader to \cite[Chap 2]{heunen2019categories}. So for category $\mathbb{X}$ that has finite biproducts, we denote the biproduct as $\oplus$, the projections as $\pi_j: A_1 \oplus \hdots \oplus A_n \to A_j$, the injections as $\iota_j: A_j \to A_1 \oplus \hdots \oplus A_n$, the zero object as $\mathsf{0}$, the sum of maps as $f+g$, and lastly the zero maps as $0$. 

\begin{definition} \cite[Def 2.39]{heunen2019categories} A dagger category $(\mathbb{X}, \dagger)$ has \textbf{finite $\dagger$-biproducts} if $\mathbb{X}$ has finite biproducts such that the adjoints of the projections are the injections, that is, $\pi^\dagger_j = \iota_j$. 
\end{definition}

Using dagger biproducts, we may now introduce generalized SVD: 


\begin{definition} In a dagger category $(\mathbb{X}, \dagger)$ with finite $\dagger$-biproducts, a \textbf{generalized singular value decomposition} (GSVD) of a map $f: A\to B$ is a triple of maps $(u: A \to X \oplus Z, d: X \to Y, v: Y \oplus W \to B)$ such that $u$ and $v$ are unitary and $d$ is an isomorphism, and such that $f = u(d \oplus 0)v$. 
\end{definition}

\begin{lemma} In a dagger category $(\mathbb{X}, \dagger)$ with finite $\dagger$-biproducts, if for a map $f: A\to B$, we have that $(u_1: A \to X_1 \oplus Z_1, d: X_1 \to Y_1, v_1: Y_1 \oplus W_1 \to B)$ and $(u_2: A \to X_2 \oplus Z_2, d: X_2 \to Y_2, v_2: Y_2 \oplus W_2 \to B)$ are both GSVDs of $f$, then there exists unique unitary maps $x: X_1 \to X_2$, $y: Y_1 \to Y_2$, $z: Z_1 \to Z_2$, and $w: W_1 \to W_2$ such that $u_1(x \oplus z) = u_2$, $d_1 y= x d_2$, and $v_1 = (y \oplus w) v_2$. 
\end{lemma}
\begin{proof} Define $x$, $y$, $z$, and $w$ as the composites $x:= \iota_1 u^\dagger_1 u_2 \pi_1$, $y:= \iota_1 v_1 v^\dagger_2 \pi_1$, $z:=\iota_2 u^\dagger_1 u_2 \pi_2$, and lastly $w:=\iota_2 v_1 v^\dagger_2 \pi_2$. By straightforward diagram chasing, one can check all the necessary identities. 
\end{proof}

We will explain below why this recaptures precisely SVD for complex matrices. We first observe that every GSVD induces a GCSVD. Therefore by applying the results of the previous section, having a GSVD implies that we have a M-P inverse: 

\begin{proposition}
In a dagger category $(\mathbb{X}, \dagger)$ with $\dagger$-biproducts, suppose that a map $f: A\to B$ has a GSVD $(u: A \to X \oplus Z, d: X \to Y, v: Y \oplus W \to B)$. Then $(u\pi_1: A \to X, d: X \to Y, v\iota_1: Y \to B)$ is a GCSVD of $f$, and therefore $f$ is M-P split where $f^\circ := v^\dagger (d^{\text{-}1} \oplus 0) u^\dagger$.
\end{proposition} 
\begin{proof} A unitary composed with a (co)isometry is always a (co)isometry. So $u\pi_1$ is a coisometry and $\iota_1 v$ is an isometry. Next, note that the $\dagger$-biproduct structure gives us that $a \oplus b = \pi_1 a \iota_1 + \pi_2 a \iota_2$. So in our case, we have that $d \oplus 0 = \pi_1 d \iota_1$. Therefore, we have that $f= u \pi_1 d \iota_1 v$. So we conclude that $(u \pi_1, d, \iota_1 v)$ is a GCSVD of $f$. Applying Proposition \ref{lem:CSVDMP} we get that $f^\circ := v^\dagger \pi_1 d^{-1} \iota_1 u^\dagger$, which can alternatively be written as $f^\circ := v^\dagger (d^{\text{-}1} \oplus 0) u^\dagger$. 
\end{proof}

Let us explain how GSVD does indeed generalize how SVD is used to compute M-P inverses for matrices. As explained in \cite[Sec 2.2.4]{heunen2019categories}, in a dagger category with finite dagger biproducts, a map $F: A_1 \oplus \hdots \oplus A_n \to B_1 \oplus \hdots \oplus B_m$ is uniquely determined by a family of maps $f_{i,j}: A_i \to B_j$. Therefore $F$ can be represented as a $n\times m$ matrix where the term in the $i$-th row and $j$-th column is $f_{i,j}$. So if $f$ has a GSVD $(u, d, v)$, we may expand $d \oplus 0$ as a $2 \times 2$ matrix, and therefore write $f$ and $f^\circ$ as: 
\begin{align*}
    f = u \begin{bmatrix} d & 0 \\
    0 & 0 
    \end{bmatrix} v && f^\circ = v^\dagger \begin{bmatrix} d^{\text{-}1} & 0 \\
    0 & 0 
    \end{bmatrix} u^\dagger 
\end{align*}
which recaptures precisely how M-P inverses were constructed using SVD in Example \ref{ex:MPmatC}. We now wish to go in the other direction, that is, going from a M-P inverse to a GSVD. To do so, we will need to use dagger kernels. For a refresher on ordinary kernels, we refer the reader to \cite[Sec 2.4.2]{heunen2019categories}. 

\begin{definition} \cite[Def 2.1]{heunen2010quantum} In a dagger category $(\mathbb{X}, \dagger)$ with a zero object, a map $f: A \to B$ has a \textbf{$\dagger$-kernel} if $f$ has a kernel $k: \mathsf{ker}(f) \to A$ such that $k$ is an isometry. A \textbf{dagger kernel category} is a dagger category with a zero object such that every map has a dagger kernel.
\end{definition}

In \cite{robinson1987generalized}, Puystjens and Robinson describe many necessary and sufficient conditions for when a map that has a kernel has a M-P inverse in a dagger category which is enriched over Abelian groups. However, dagger kernels are not discussed in \cite{robinson1987generalized}. Therefore, one could specialize certain results in \cite{robinson1987generalized} for dagger kernels instead. In this paper, we will show that having a M-P inverse \emph{and} a dagger kernel is equivalent to having a GSVD. Also note that, unlike in \cite{robinson1987generalized}, we do not assume that we are working in a setting with negatives (i.e. additive inverses). Because of this, the statement does require a modest extra compatibility condition between the M-P inverse and the dagger kernel. 

\begin{proposition} In a dagger category $(\mathbb{X}, \dagger)$ with $\dagger$-biproducts, a map $f$ has a GSVD if and only if $f$ is M-P split and $f$ has a $\dagger$-kernel $k: \mathsf{ker}(f) \to A$ and $f^\dagger$ has a $\dagger$-kernel $c: \mathsf{ker}(f^\dagger) \to B$ such that $f f^\circ + k^\dagger k = 1_A$ and $f^\circ f+ c^\dagger c = 1_B$. 
\end{proposition} 
\begin{proof} Suppose that $f: A\to B$ has a GSVD $(u: A \to X \oplus Z, d: X \to Y, v: Y \oplus W \to B)$. We have already explained why $f$ is M-P split in the above lemma. Using the $\dagger$-biproduct identity that $\iota_i \pi_j = 0$ if $i\neq j$ and $\iota_j \pi_j = 1$, it is straightforward to check that $\iota_2 u^\dagger: Z \to A$ is a $\dagger$-kernel of $f$ and that $\iota_2 v: W \to B$ is a $\dagger$-kernel of $f^\dagger$. For the extra identities, we first note that $ff^\circ = u(1_X \oplus 0)u^\dagger$ and $f^\circ f = v^\dagger(1_Y \oplus 0)v$, which we can alternatively write as $ff^\circ = u\pi_1 \iota_1u^\dagger$ and $f^\circ f = v^\dagger\pi_1 \iota_1v$. Then using the other $\dagger$-biproduct identity that $\pi_1 \iota_1 + \pi_2 \iota_2 = 1$, it follows that $f f^\circ + k^\dagger k  =1_A$ and $f^\circ f+ c^\dagger c  =1_B$ as desired. 

Conversely, suppose that $f$ is M-P split, and has a $\dagger$-kernel $k: \mathsf{ker}(f) \to A$ and $f^\dagger$ has a $\dagger$-kernel $c: \mathsf{ker}(f^\dagger) \to B$ such that the two equalities $f f^\circ + k^\dagger k  =1_A$ and $f^\circ f+ c^\dagger c  =1_B$ also hold. Then by Prop \ref{lem:CSVDMP}, $f$ also has a GCSVD $(r: A \to X, d: X \to Y, s: Y \to B)$, so, in particular, $f = r d s$ and $d$ is an isomorphism. Then, using matrix notation, define $u: A \to X \oplus \mathsf{ker}(f)$ and $v: Y \oplus \mathsf{ker}(f^\dagger) \to B$ respectively as $u := \begin{bmatrix} r & k^\dagger \end{bmatrix}$ and $v := \begin{bmatrix} s \\ c \end{bmatrix}$. We first compute that: 
\begin{align*}
  u(d \oplus 0)v = \begin{bmatrix} r & k^\dagger \end{bmatrix}\begin{bmatrix} d & 0 \\
  0 & 0 \end{bmatrix}  \begin{bmatrix} s \\ c \end{bmatrix} = \begin{bmatrix} rd & 0\end{bmatrix}\begin{bmatrix} s \\ c \end{bmatrix} = r d s = f
\end{align*}
So $f = u(d \oplus 0)v$ as desired. We must also show that $u$ and $v$ are unitary. To show that $u$ is unitary, recall that $r r^\dagger = f^\circ f$ and also that since $krd s =kf= 0$, it follows that $kr =0$. Therefore we compute: 
\begin{gather*}
    u u^\dagger = \begin{bmatrix} r & k^\dagger \end{bmatrix} \begin{bmatrix} r & k^\dagger \end{bmatrix}^\dagger = \begin{bmatrix} r & k^\dagger \end{bmatrix} \begin{bmatrix} r^\dagger \\ k \end{bmatrix} = rr^\dagger + k^\dagger k  = ff^\circ+ k^\dagger k = 1_A \\
     u^\dagger u = \begin{bmatrix} r^\dagger \\ k \end{bmatrix} \begin{bmatrix} r & k^\dagger \end{bmatrix} = \begin{bmatrix} r^\dagger r & r^\dagger k^\dagger \\
     kr & kk^\dagger \end{bmatrix} = \begin{bmatrix} 1_A & 0 \\
     0 & 1_{\mathsf{ker}(f)}\end{bmatrix} = 1_{X \oplus \mathsf{ker}(f)}
    \end{gather*}
So $u$ is unitary. Similarly, we can show that $v$ is unitary. So we conclude that $(u,d, v)$ is a GSVD of $f$. 
\end{proof}

\begin{corollary}  In a dagger category $(\mathbb{X}, \dagger)$ with finite $\dagger$-biproducts and negatives, a map $f$ has a GSVD if and only if $f$ is M-P split and both $f$ and $f^\dagger$ have $\dagger$-kernels. 
\end{corollary}
\begin{proof} We need only show that $f f^\circ + k^\dagger k  =1_A$ and $f^\circ f+ c^\dagger c  =1_B$. First note that $(1_A - f f^\circ)f=0$ and $(1_B - f^\circ f)f^\dagger =0$ (the latter of which is by Lemma \ref{lem:MPinv1}.(\ref{lem:MPfdageq})). So by universal property of the kernel, there exist unique maps $z_1$ and $z_2$ such that $z_1k = 1_A - f f^\circ$ and $z_2c = 1_B - f^\circ f$. By post-composing by $k^\dagger$ and $c^\dagger$ respectively, and also by using that $f^\circ k^\dagger =0$ (which follows from  Lemma \ref{lem:MPinv1}.(\ref{lem:MPfcirceq})) and $f c^\dagger =0$, we then obtain that $z_1 = k^\dagger$ and $z_2=c^\dagger$. Therefore, $k^\dagger k= 1_A - f f^\circ$ and $c^\dagger c= 1_B - f^\circ f$, which in turn implies the desired equalities. 
\end{proof}

Therefore, in a setting with negatives and all dagger kernels, we may state that: 

\begin{corollary} In a dagger kernel category $(\mathbb{X}, \dagger)$ with finite $\dagger$-biproducts and negatives, a map $f$ has a GSVD if and only if $f$ is M-P split. 
\end{corollary}

Finally, assuming also that we are in a dagger idempotent complete setting, we obtain a precise characterization of M-P invertible maps in terms of a generalized version of SVD: 

\begin{theorem}\label{thm:MPSVD} In a dagger kernel category $(\mathbb{X}, \dagger)$ that is $\dagger$-idempotent complete and which has finite $\dagger$-biproducts and negatives, a map $f$ is M-P invertible if and only if $f$ has a GSVD. 
\end{theorem}

\section{Polar Decomposition}\label{sec:CPD}

It is straightforward to give a generalized version of polar decomposition (PD) in a dagger category, it is the statement that a map factorizes as a partial isometry followed by a positive map. However, the statement of PD for bounded linear maps between Hilbert spaces is stronger: it also involves a requirement on the kernel (or range) of the partial isometry. In \cite[Thm 8.3]{higham2008functions}, Higham nicely explains how M-P inverses can play a role in the PD of complex matrices and can be used to replace that extra requirement. So recall that for an $n \times m$ complex matrix, $A$, there exists a unique partial isometry $U$ and unique a positive semi-definite Hermitian matrix $H$ such that $A=UH$ and $\mathsf{range}(U^\dagger) = \mathsf{range}(H)$. The matrix $H$ is given by the square root of the matrix $A^\dagger A$, so $H = (A^\dagger A)^{\frac{1}{2}}$, while the matrix $U$ is constructed using the M-P inverse of $H$, so $U = A H^\circ$. Furthermore, the condition $\mathsf{range}(U^\ast) = \mathsf{range}(H)$ can be equivalently described in terms of M-P inverses as the equality $U^\dagger U = H H^\circ$ (where note that $U^\circ = U^\dagger$ since $U$ is a partial isometry). Therefore PD of complex matrices can be completely expressed in terms of M-P inverses. As such in this section, we introduce the notion of a Moore-Penrose polar decomposition of maps in an arbitrary dagger category, which recaptures precisely PD for complex matrices. 

\begin{definition} In a dagger category $(\mathbb{X}, \dagger)$, for a map $f: A \to B$,
\begin{enumerate}[{\em (i)}]
\item A \textbf{generalized polar decomposition} (GPD) of $f$ is a pair of maps $(u: A \to B, h: B \to B)$ where $u$ is a partial isometry and $h$ is a positive map such that $f=uh$;
\item A \textbf{Moore-Penrose polar decomposition} (M-P PD) of $f$ is a GPD $(u: A \to B, h: B \to B)$ of $f$ such that $h$ is M-P invertible and $u^\dagger u = h h^\circ$.
\end{enumerate}
\end{definition}

We will show that for $f$ to have a M-P PD is equivalent to requiring that $f$ be  M-P invertible and $f^\dagger f$ has a square-root. The following definition is a Moore-Penrose version of Selinger's definition \cite[Def 5.13]{selinger2008idempotents}. 

\begin{definition} In a dagger category $(\mathbb{X}, \dagger)$, a M-P invertible positive map $p: A \to A$ has a \textbf{Moore-Penrose square root} (M-P square root) if there exists a M-P invertible positive map $\sqrt{p}: A \to A$ such that $\sqrt{p} \sqrt{p} = p$. A dagger category is said to have \textbf{(unique) M-P square roots} if all M-P invertible positive maps have a (unique) M-P square root. 
\end{definition}

\begin{proposition}\label{prop:MPPD} In a dagger category $(\mathbb{X}, \dagger)$, a map $f$ has a M-P PD if and only if $f$ is M-P invertible and $f^\dagger f$ has a M-P square root. 
\end{proposition} 
\begin{proof} Suppose that $(u: A \to B, h: B \to B)$ is a M-P PD of a map $f: A \to B$. Since $u$ is a partial isometry, $u$ is M-P invertible where $u^\circ = u^\dagger$. We also have that since $h$ is positive, it is self-dual $h^\dagger =h$, and by Lemma \ref{lem:MPinv1}.(\ref{lem:MPselfadjoint}) we have that $h^\circ h = h h^\circ$. Therefore by the assumption that $u^\dagger u = h h^\circ$, it easily follows that $u^\dagger u h h^\circ = h h^\circ u^\dagger u$. Therefore by Lemma \ref{lem:MPcomp}.(ii), we have that $f =uh$ is M-P invertible whose M-P inverse is $f^\circ = (uh)^\circ = h^\circ u^\dagger$. By Lemma \ref{lem:MPinv1}.(\ref{lem:MPdagidem}), $f^\dagger f$ is a M-P invertible positive map. So it remains to compute that:
\[h h =  h h h^\circ h = h u^\dagger u h = (uh)^\dagger uh = f^\dagger f\] 
So $hh =f^\dagger f$, and therefore $h$ is a M-P square root of $f^\dagger f$.

Conversely, suppose that $f$ is M-P invertible and $f^\dagger f$ has a M-P square root $\sqrt{f^\dagger f}$. So define $h: B \to B$ as $h := \sqrt{f^\dagger f}$, and define $u: A \to B$ as the composite $u := fh^\circ$. We then compute the following: 
\begin{gather*}
    u u^\dagger u = fh^\circ (fh^\circ)^\dagger fh^\circ = fh^\circ h^\circ f^\dagger f h^\circ = fh^\circ h^\circ h h h^\circ = fh^\circ h h^\circ  h h^\circ = f h^\circ h h^\circ = f h^\circ = u \\
     u^\dagger u =  (fh^\circ)^\dagger fh^\circ = h^\circ f^\dagger f h^\circ = h^\circ h h h^\circ = h h^\circ h h^\circ = h h^\circ \\ 
        uh = f h^\circ h = fh^\circ h h^\circ h = f h^\circ h^\circ h h = f (h h)^\circ h h =
   f (f^\dagger f)^\circ f^\dagger f = f f^\circ f = f 
\end{gather*}
So $u$ is an isometry, $ u^\dagger u=h h^\circ$, and $f=uh$. So we conclude that $(u,h)$ is a M-P PD of $f$. 
\end{proof}

\begin{corollary}  In a dagger category $(\mathbb{X}, \dagger)$ with M-P square roots, a map $f$ is M-P invertible if and only if $f$ has a M-P PD. 
\end{corollary}

Unlike PD which is always unique, M-P PD is not necessarily unique in an arbitrary dagger category. The reason PD is unique is due to the fact that positive semi-definite Hermitian matrices have unique square roots. Therefore, if we work in a dagger category where the positive maps do have unique square roots, then M-P PD is also unique as desired. 

\begin{lemma} In a dagger category $(\mathbb{X}, \dagger)$ with unique M-P square roots, M-P PDs are unique. 
\end{lemma}
\begin{proof} Suppose that ${(u: A \to B, h: B \to B)}$ and ${(v: A \to B, k: B \to B)}$ are both M-P PDs of a map $f: A \to B$. By Proposition \ref{prop:MPPD}, we have that $u=fh^\circ$ and $v=fk^\circ$, and that $h$ and $k$ are positive maps such that $hh=f^\dagger f=kk$. By the uniqueness of M-P square roots, this implies that $h=k$. In turn, this also implies that $u = v$. So we conclude that a M-P PD is unique. 
\end{proof}

\section{Conclusion}

In this paper, we revisited and added to the story of Moore-Penrose inverses in a dagger category. This work was motivated in part by wishing to understand how partial isomorphisms (in the restriction categories sense) generalize to dagger categories: Moore-Penrose inverses seem to provide the appropriate generalization. However, their theory is more sophisticated and could be better understood. In particular, although a start has been made here, there is more to be understood about their compositional behaviour and their relation to dagger idempotents. Moore-Penrose inverses should also be considered in relation to other dagger structures, such as dagger limits \cite{heunen2018limits}, dagger monads \cite{heunen2016monads}, and dagger compact closedness \cite{selinger2007dagger}. One should also find other interesting examples of Moore-Penrose dagger categories. We conjecture that certain fragments of the ZX-calculus \cite{coecke2018picturing} and possibly PROPs with weights on strings will be Moore-Penrose dagger categories. Finally, as the Moore-Penrose inverse has many practical applications, it would also be worthwhile generalizing these applications to Moore-Penrose dagger categories. This may in turn lead to further applications for Moore-Penrose inverses. 



\nocite{*}
\bibliographystyle{eptcs}
\bibliography{generic}

\begin{thebibliography}{10}
\providecommand{\bibitemdeclare}[2]{}
\providecommand{\surnamestart}{}
\providecommand{\surnameend}{}
\providecommand{\urlprefix}{Available at }
\providecommand{\url}[1]{\texttt{#1}}
\providecommand{\href}[2]{\texttt{#2}}
\providecommand{\urlalt}[2]{\href{#1}{#2}}
\providecommand{\doi}[1]{doi:\urlalt{https://doi.org/#1}{#1}}
\providecommand{\eprint}[1]{arXiv:\urlalt{https://arxiv.org/abs/#1}{#1}}
\providecommand{\bibinfo}[2]{#2}

\bibitemdeclare{article}{baksalary2021moore}
\bibitem{baksalary2021moore}
\bibinfo{author}{O.~M. \surnamestart Baksalary\surnameend} \&
  \bibinfo{author}{G.~\surnamestart Trenkler\surnameend}
  (\bibinfo{year}{2021}): \emph{\bibinfo{title}{The Moore--Penrose inverse: a
  hundred years on a frontline of physics research}}.
\newblock {\slshape \bibinfo{journal}{The European Physical Journal H}}
  \bibinfo{volume}{46}, pp. \bibinfo{pages}{1--10},
  \doi{10.1140/epjh/s13129-021-00011-y}.

\bibitemdeclare{article}{bapat1990generalized}
\bibitem{bapat1990generalized}
\bibinfo{author}{R.~B. \surnamestart Bapat\surnameend}, \bibinfo{author}{K.~P.
  S.~B. \surnamestart Rao\surnameend} \& \bibinfo{author}{K.~M. \surnamestart
  Prasad\surnameend} (\bibinfo{year}{1990}): \emph{\bibinfo{title}{Generalized
  inverses over integral domains}}.
\newblock {\slshape \bibinfo{journal}{Linear Algebra and its Applications}}
  \bibinfo{volume}{140}, pp. \bibinfo{pages}{181--196},
  \doi{10.1016/0024-3795(90)90229-6}.

\bibitemdeclare{article}{bapat1992moore}
\bibitem{bapat1992moore}
\bibinfo{author}{R.~B. \surnamestart Bapat\surnameend} \&
  \bibinfo{author}{D.~W. \surnamestart Robinson\surnameend}
  (\bibinfo{year}{1992}): \emph{\bibinfo{title}{The Moore-Penrose inverse over
  a commutative ring}}.
\newblock {\slshape \bibinfo{journal}{Linear algebra and its applications}}
  \bibinfo{volume}{177}, pp. \bibinfo{pages}{89--103},
  \doi{10.1016/0024-3795(92)90318-5}.

\bibitemdeclare{article}{ben2002moore}
\bibitem{ben2002moore}
\bibinfo{author}{A.~\surnamestart Ben-Israel\surnameend}
  (\bibinfo{year}{2002}): \emph{\bibinfo{title}{The Moore of the Moore-Penrose
  Inverse}}.
\newblock {\slshape \bibinfo{journal}{The Electronic Journal of Linear
  Algebra}} \bibinfo{volume}{9}, pp. \bibinfo{pages}{150--157},
  \doi{10.13001/1081-3810.1083}.

\bibitemdeclare{book}{campbell2009generalized}
\bibitem{campbell2009generalized}
\bibinfo{author}{S.~\surnamestart Campbell\surnameend} \&
  \bibinfo{author}{C.~\surnamestart Meyer\surnameend} (\bibinfo{year}{2009}):
  \emph{\bibinfo{title}{Generalized inverses of linear transformations}}.
\newblock \bibinfo{publisher}{SIAM}, \doi{10.1137/1.9780898719048}.

\bibitemdeclare{article}{cockett2002restriction}
\bibitem{cockett2002restriction}
\bibinfo{author}{R.~\surnamestart Cockett\surnameend} \&
  \bibinfo{author}{S.~\surnamestart Lack\surnameend} (\bibinfo{year}{2002}):
  \emph{\bibinfo{title}{Restriction categories I: categories of partial maps}}.
\newblock {\slshape \bibinfo{journal}{Theoretical computer science}}
  \bibinfo{volume}{270}(\bibinfo{number}{1-2}), pp. \bibinfo{pages}{223--259},
  \doi{10.1016/S0304-3975(00)00382-0}.

\bibitemdeclare{inproceedings}{coecke2018picturing}
\bibitem{coecke2018picturing}
\bibinfo{author}{B.~\surnamestart Coecke\surnameend} \&
  \bibinfo{author}{A.~\surnamestart Kissinger\surnameend}
  (\bibinfo{year}{2018}): \emph{\bibinfo{title}{Picturing quantum processes: A
  first course on quantum theory and diagrammatic reasoning}}.
\newblock In: {\slshape \bibinfo{booktitle}{Diagrammatic Representation and
  Inference: 10th International Conference, Diagrams 2018, Edinburgh, UK, June
  18-22, 2018, Proceedings 10}}, \bibinfo{organization}{Springer}, pp.
  \bibinfo{pages}{28--31}, \doi{10.1017/9781316219317}.

\bibitemdeclare{inproceedings}{drazin1979regular}
\bibitem{drazin1979regular}
\bibinfo{author}{M.~P. \surnamestart Drazin\surnameend} (\bibinfo{year}{1979}):
  \emph{\bibinfo{title}{Regular semigroups with involution}}.
\newblock In: {\slshape \bibinfo{booktitle}{Proc. Symp. on Regular
  Semigroups}}, pp. \bibinfo{pages}{29--46}.

\bibitemdeclare{book}{freyd1990categories}
\bibitem{freyd1990categories}
\bibinfo{author}{P.~J. \surnamestart Freyd\surnameend} \&
  \bibinfo{author}{A.~\surnamestart Scedrov\surnameend} (\bibinfo{year}{1990}):
  \emph{\bibinfo{title}{Categories, Allegories}}.
\newblock \bibinfo{publisher}{Elsevier}.

\bibitemdeclare{inproceedings}{gumm2014coalgebraic}
\bibitem{gumm2014coalgebraic}
\bibinfo{author}{H.~P. \surnamestart Gumm\surnameend} \&
  \bibinfo{author}{M.~\surnamestart Zarrad\surnameend} (\bibinfo{year}{2014}):
  \emph{\bibinfo{title}{Coalgebraic simulations and congruences}}.
\newblock In: {\slshape \bibinfo{booktitle}{Coalgebraic Methods in Computer
  Science: 12th IFIP WG 1.3 International Workshop, CMCS 2014, Colocated with
  ETAPS 2014, Grenoble, France, April 5-6, 2014, Revised Selected Papers}},
  \bibinfo{organization}{Springer}, pp. \bibinfo{pages}{118--134},
  \doi{10.1007/978-3-662-44124-4_7}.

\bibitemdeclare{book}{hagen2000c}
\bibitem{hagen2000c}
\bibinfo{author}{R.~\surnamestart Hagen\surnameend},
  \bibinfo{author}{S.~\surnamestart Roch\surnameend} \&
  \bibinfo{author}{B.~\surnamestart Silbermann\surnameend}
  (\bibinfo{year}{2000}): \emph{\bibinfo{title}{C*-algebras and numerical
  analysis}}.
\newblock \bibinfo{publisher}{CRC Press}, \doi{10.1201/9781482270679}.

\bibitemdeclare{article}{heunen2010quantum}
\bibitem{heunen2010quantum}
\bibinfo{author}{C.~\surnamestart Heunen\surnameend} \&
  \bibinfo{author}{B.~\surnamestart Jacobs\surnameend} (\bibinfo{year}{2010}):
  \emph{\bibinfo{title}{Quantum logic in dagger kernel categories}}.
\newblock {\slshape \bibinfo{journal}{Order}}
  \bibinfo{volume}{27}(\bibinfo{number}{2}), pp. \bibinfo{pages}{177--212},
  \doi{10.1016/j.entcs.2011.01.024}.

\bibitemdeclare{article}{heunen2016monads}
\bibitem{heunen2016monads}
\bibinfo{author}{C.~\surnamestart Heunen\surnameend} \&
  \bibinfo{author}{M.~\surnamestart Karvonen\surnameend}
  (\bibinfo{year}{2016}): \emph{\bibinfo{title}{Monads on dagger categories}}.
\newblock {\slshape \bibinfo{journal}{Theory and Applications of Categories}}
  \bibinfo{volume}{31}(\bibinfo{number}{35}), pp. \bibinfo{pages}{1016--1043},
  \doi{10.48550/arXiv.1602.04324}.

\bibitemdeclare{article}{heunen2018limits}
\bibitem{heunen2018limits}
\bibinfo{author}{C.~\surnamestart Heunen\surnameend} \&
  \bibinfo{author}{M.~\surnamestart Karvonen\surnameend}
  (\bibinfo{year}{2019}): \emph{\bibinfo{title}{Limits in dagger categories}}.
\newblock {\slshape \bibinfo{journal}{Theory and Applications of Categories}}
  \bibinfo{volume}{34}(\bibinfo{number}{18}), pp. \bibinfo{pages}{468--513},
  \doi{10.48550/arXiv.1803.06651}.

\bibitemdeclare{book}{heunen2019categories}
\bibitem{heunen2019categories}
\bibinfo{author}{C.~\surnamestart Heunen\surnameend} \&
  \bibinfo{author}{J.~\surnamestart Vicary\surnameend} (\bibinfo{year}{2019}):
  \emph{\bibinfo{title}{Categories for Quantum Theory: an introduction}}.
\newblock \bibinfo{publisher}{Oxford University Press},
  \doi{10.1093/oso/9780198739623.001.0001}.

\bibitemdeclare{book}{higham2008functions}
\bibitem{higham2008functions}
\bibinfo{author}{N.~J. \surnamestart Higham\surnameend} (\bibinfo{year}{2008}):
  \emph{\bibinfo{title}{Functions of matrices: theory and computation}}.
\newblock \bibinfo{publisher}{SIAM}, \doi{10.1137/1.9780898717778}.

\bibitemdeclare{article}{huylebrouck1984moore}
\bibitem{huylebrouck1984moore}
\bibinfo{author}{D.~\surnamestart Huylebrouck\surnameend},
  \bibinfo{author}{R.~\surnamestart Puystjens\surnameend} \&
  \bibinfo{author}{J.~\surnamestart Van~Geel\surnameend}
  (\bibinfo{year}{1984}): \emph{\bibinfo{title}{The Moore-Penrose inverse of a
  matrix over a semi-simpie artinian ring}}.
\newblock {\slshape \bibinfo{journal}{Linear and Multilinear Algebra}}
  \bibinfo{volume}{16}(\bibinfo{number}{1-4}), pp. \bibinfo{pages}{239--246},
  \doi{10.1080/03081088408817625}.

\bibitemdeclare{article}{moore1920reciprocal}
\bibitem{moore1920reciprocal}
\bibinfo{author}{E.~H. \surnamestart Moore\surnameend} (\bibinfo{year}{1920}):
  \emph{\bibinfo{title}{On the reciprocal of the general algebraic matrix}}.
\newblock {\slshape \bibinfo{journal}{Bull. Am. Math. Soc.}}
  \bibinfo{volume}{26}, pp. \bibinfo{pages}{394--395},
  \doi{10.1090/S0002-9904-1920-03322-7}.

\bibitemdeclare{article}{pearl1968generalized}
\bibitem{pearl1968generalized}
\bibinfo{author}{M.~H. \surnamestart Pearl\surnameend} (\bibinfo{year}{1968}):
  \emph{\bibinfo{title}{Generalized inverses of matrices with entries taken
  from an arbitrary field}}.
\newblock {\slshape \bibinfo{journal}{Linear Algebra and its Applications}}
  \bibinfo{volume}{1}(\bibinfo{number}{4}), pp. \bibinfo{pages}{571--587},
  \doi{10.1016/0024-3795(68)90028-1}.

\bibitemdeclare{inproceedings}{penrose1955generalized}
\bibitem{penrose1955generalized}
\bibinfo{author}{R.~\surnamestart Penrose\surnameend} (\bibinfo{year}{1955}):
  \emph{\bibinfo{title}{A generalized inverse for matrices}}.
\newblock In: {\slshape \bibinfo{booktitle}{Mathematical proceedings of the
  Cambridge philosophical society}}, \bibinfo{volume}{51},
  \bibinfo{organization}{Cambridge University Press}, pp.
  \bibinfo{pages}{406--413}, \doi{10.1017/S0305004100030401}.

\bibitemdeclare{article}{puystjens1981moore}
\bibitem{puystjens1981moore}
\bibinfo{author}{R.~\surnamestart Puystjens\surnameend} \&
  \bibinfo{author}{D.~w. \surnamestart Robinson\surnameend}
  (\bibinfo{year}{1981}): \emph{\bibinfo{title}{The Moore-Penrose inverse of a
  morphism with factorization}}.
\newblock {\slshape \bibinfo{journal}{Linear Algebra and its Applications}}
  \bibinfo{volume}{40}, pp. \bibinfo{pages}{129--141},
  \doi{10.1016/0024-3795(81)90145-2}.

\bibitemdeclare{article}{puystjens1984moore}
\bibitem{puystjens1984moore}
\bibinfo{author}{R.~\surnamestart Puystjens\surnameend} \&
  \bibinfo{author}{D.~W. \surnamestart Robinson\surnameend}
  (\bibinfo{year}{1984}): \emph{\bibinfo{title}{The Moore-Penrose inverse of a
  morphism in an additive category}}.
\newblock {\slshape \bibinfo{journal}{Communications in algebra}}
  \bibinfo{volume}{12}(\bibinfo{number}{3}), pp. \bibinfo{pages}{287--299},
  \doi{10.1080/00927878408823004}.

\bibitemdeclare{article}{robinson1985ep}
\bibitem{robinson1985ep}
\bibinfo{author}{R.~\surnamestart Puystjens\surnameend} \&
  \bibinfo{author}{D.~W. \surnamestart Robinson\surnameend}
  (\bibinfo{year}{1985}): \emph{\bibinfo{title}{EP morphisms}}.
\newblock {\slshape \bibinfo{journal}{Linear algebra and its applications}}
  \bibinfo{volume}{64}, pp. \bibinfo{pages}{157--174},
  \doi{10.1016/0024-3795(85)90273-3}.

\bibitemdeclare{article}{robinson1987generalized}
\bibitem{robinson1987generalized}
\bibinfo{author}{R.~\surnamestart Puystjens\surnameend} \&
  \bibinfo{author}{D.~W. \surnamestart Robinson\surnameend}
  (\bibinfo{year}{1987}): \emph{\bibinfo{title}{Generalized inverses of
  morphisms with kernels}}.
\newblock {\slshape \bibinfo{journal}{Linear Algebra and Its Applications}}
  \bibinfo{volume}{96}, pp. \bibinfo{pages}{65--86},
  \doi{10.1016/0024-3795(87)90336-3}.

\bibitemdeclare{article}{puystjens1990symmetric}
\bibitem{puystjens1990symmetric}
\bibinfo{author}{R.~\surnamestart Puystjens\surnameend} \&
  \bibinfo{author}{D.~W. \surnamestart Robinson\surnameend}
  (\bibinfo{year}{1990}): \emph{\bibinfo{title}{Symmetric morphisms and the
  existence of Moore-Penrose inverses}}.
\newblock {\slshape \bibinfo{journal}{Linear Algebra and Its Applications}}
  \bibinfo{volume}{131}, pp. \bibinfo{pages}{51--69},
  \doi{10.1016/0024-3795(90)90374-L}.

\bibitemdeclare{article}{rao1975generalized}
\bibitem{rao1975generalized}
\bibinfo{author}{P.~S. S. N. V.~P. \surnamestart Rao\surnameend} \&
  \bibinfo{author}{K.~P. S.~B. \surnamestart Rao\surnameend}
  (\bibinfo{year}{1975}): \emph{\bibinfo{title}{On Generalized Inverses of
  Boolean Matrices}}.
\newblock {\slshape \bibinfo{journal}{Linear Algebra and its applications}}
  \bibinfo{volume}{11}(\bibinfo{number}{2}), pp. \bibinfo{pages}{135--153},
  \doi{10.1016/0024-3795(75)90054-3}.

\bibitemdeclare{article}{selinger2007dagger}
\bibitem{selinger2007dagger}
\bibinfo{author}{P.~\surnamestart Selinger\surnameend} (\bibinfo{year}{2007}):
  \emph{\bibinfo{title}{Dagger compact closed categories and completely
  positive maps}}.
\newblock {\slshape \bibinfo{journal}{Electronic Notes in Theoretical computer
  science}} \bibinfo{volume}{170}, pp. \bibinfo{pages}{139--163},
  \doi{10.1016/j.entcs.2006.12.018}.

\bibitemdeclare{article}{selinger2008idempotents}
\bibitem{selinger2008idempotents}
\bibinfo{author}{P.~\surnamestart Selinger\surnameend} (\bibinfo{year}{2008}):
  \emph{\bibinfo{title}{Idempotents in dagger categories}}.
\newblock {\slshape \bibinfo{journal}{Electronic Notes in Theoretical Computer
  Science}} \bibinfo{volume}{210}, pp. \bibinfo{pages}{107--122},
  \doi{10.1016/j.entcs.2008.04.021}.

\end{thebibliography}
\end{document}